\documentclass[submission,copyright,creativecommons]{eptcs}
\providecommand{\event}{QPL 2011}
\usepackage{breakurl}

\usepackage{amsmath,amssymb,amsthm}
\newtheorem{theorem}{Theorem}
\newtheorem{lemma}[theorem]{Lemma}
\newtheorem{proposition}[theorem]{Proposition}
\newtheorem{corollary}[theorem]{Corollary}

\theoremstyle{definition}

\newtheorem{evidence}[theorem]{Numerical evidence}
\theoremstyle{remark}

\usepackage{xypic}
\usepackage{graphicx}

\newcommand{\bra}[1]{\left\langle{#1}\right\vert}
\newcommand{\ket}[1]{\left\vert{#1}\right\rangle}
\newcommand{\udots}{\mathinner{\mskip1mu\raise1pt\vbox{\kern7pt\hbox{.}}
   \mskip2mu\raise4pt\hbox{.}\mskip2mu\raise7pt\hbox{.}\mskip1mu}}

\usepackage{tikz}
\usetikzlibrary{decorations.markings}
\pgfdeclarelayer{edgelayer}
\pgfdeclarelayer{nodelayer}
\pgfsetlayers{edgelayer,nodelayer,main}
\tikzset{bend angle=45}
\tikzset{baseline=(current bounding box.center)}
\tikzset{column sep=2ex}
\tikzset{row sep=2ex}
\tikzstyle{none}=[inner sep=1pt]
\tikzstyle{None}=[circle, fill=white, inner sep=0pt]
\tikzstyle{box}=[draw=black, fill=white, inner sep=.3ex, rounded corners=.1ex]
\tikzstyle{dot}=[circle, draw=black, fill=black!50, inner sep=1.25pt]
\tikzstyle{cross}=[preaction={draw=white, -, line width=3pt}]
\tikzstyle{arrow}=[postaction=decorate]
\newcommand{\markat}{0.5}
\newcommand{\markwithsym}{>}
\newcommand{\markwith}{{\arrow[black]{\markwithsym}}}
\pgfkeys{/tikz/.cd, markat/.store in=\markat, markwith/.store in=\markwithsym} 
\tikzset{decoration={markings, mark=at position \markat with \markwith}}

\usepackage{bm}
\newcommand{\after}{\circ}
\newcommand{\cat}[1]{\ensuremath{\mathbf{#1}}}
\newcommand{\Cat}[1]{\ensuremath{\mathbf{#1}}}
\newcommand{\id}[1][]{\ensuremath{\mathrm{id}_{#1}}}

\newcommand{\field}[1]{\ensuremath{\mathbb{#1}}}
\newcommand{\tensor}{\ensuremath{\otimes}}
\newcommand{\ie}{\textit{i.e.}~}
\newcommand{\eg}{\textit{e.g.}~}

\newcommand{\CPM}{\ensuremath{\mathrm{CPM}}}
\newcommand{\cp}[1]{\ensuremath{\bm{#1}}}
\newcommand{\inprod}[2]{\ensuremath{\langle #1\,|\,#2 \rangle}}
\newcommand{\Tr}{\ensuremath{\mathrm{Tr}}}
\newcommand{\rank}{\ensuremath{\mathrm{rank}}}
\newcommand{\opnorm}[1]{\ensuremath{\|#1\|_{\mathrm{op}}}}
\newcommand{\swapmor}{\ensuremath{\sigma}}

\title{Completely positive classical structures \\ and sequentializable quantum protocols}
\author{Chris Heunen\thanks{Supported by the Netherlands Organisation
    for Scientific Research (NWO).} \institute{University of Oxford} \email{heunen@cs.ox.ac.uk} 
  \and Sergio Boixo\thanks{Supported by FIS2008-01236 and Defense Advanced
Research Projects Agency award N66001-09-1-2101.} \institute{USC \& Harvard University} \email{sergio@boixo.com} }
\date{}
\def\titlerunning{Completely positive classical structures and
  sequentializable quantum protocols}
\def\authorrunning{C. Heunen \& S. Boixo}
\begin{document}
\maketitle

\begin{abstract}
  We study classical structures in various categories of completely positive morphisms: on sets   
  and relations, on cobordisms, on a free dagger compact category, and on Hilbert
  spaces. As an application, we prove that quantum maps with commuting
  Kraus operators can be sequentialized. Hence such protocols are 
  precisely as robust under general dephasing noise when entangled as when sequential.
\end{abstract}

\section{Introduction}

There are two ways to model classical information in categorical
quantum mechanics. Originally, biproducts were used to capture
classical information external to the category at
hand~\cite{abramskycoecke:categoricalsemantics}. Later, so-called 
classical structures emerged as a way to model classical data internal
to the category itself~\cite{coeckepaquette:practisingphysicist,
coeckepavlovicvicary:bases, abramskyheunen:hstar}.

The setting of most interest to quantum information theory is that of
completely positive maps, which was abstracted in Selinger's
CPM-construction~\cite{selinger:completelypositive}. This
construction need not preserve biproducts. If desired, biproducts have
to be freely added again. In fact, the counterexample given
in~\cite{selinger:completelypositive} is the category of
finite-dimensional Hilbert spaces and completely positive maps, which
is of course the crucial model for quantum mechanics.

On the other hand, the CPM-construction does preserve classical
structures. However, it might introduce new classical
structures, that would therefore not model classical information. The current
paper addresses the (non)existence of such noncanonical classical
structures in categories of completely positive maps. We prove that
there are no noncanonical completely positive classical structures in
several categories: sets and relations, cobordisms, and the free
dagger compact category on one generator. 

For the main event, we then consider the category of finite-dimensional Hilbert spaces and
completely positive maps. We cannot close the question there entirely
yet. But we make enough progress to enable an application to 
quantum metrology~\cite{giovannetti_quantum_2006}. Many quantum metrology
protocols operate parallelly on a maximally entangled state. For some
protocols, such as phase estimation, frame synchronization,
and clock synchronization, there exists an equivalent
sequential version, in which entanglement is traded for repeated
operations on a simpler quantum state. We can extend the class of known
sequentializable protocols to those whose Kraus operators commute.
% The transformation between sequential and entangled protocols has also
% been used to study the class QNC of quantum circuits with
% polylogarithmic depth~\cite{moorenilsson:parallel}. This class includes standard quantum
% error-correcting encoding and decoding.

The setup of the paper is simple: Section~\ref{sec:preliminaries} introduces the necessary
ingredients, Section~\ref{sec:cpcs} considers completely positive classical
structures, and Section~\ref{sec:seq} outlines their application to
sequentializable quantum protocols. 
Much of this work has been done while both authors were at the
Institute for Quantum Information at the California Institute of
Technology. We thank Peter Selinger for pointing
out~\cite{mccurdyselinger:basicdaggercompactcategories}, and David
P{\'e}rez Garc{\'i}a, Robert K{\"o}nig, Peter Love and Spiros Michalakis for discussions.

\section{Preliminaries}\label{sec:preliminaries}

For basics on dagger compact categories we refer
to~\cite{abramskycoecke:categoricalsemantics,
coeckepaquette:practisingphysicist, selinger:graphicallanguages}.
A \emph{classical structure} in a dagger compact category
is a morphism $\delta \colon X \to
X \tensor X$ satisfying $(\id
\tensor \delta^\dag) \after (\delta \tensor \id) = \delta \after \delta^\dag = (\delta^\dag
\tensor \id) \after (\id \tensor \delta)$, $\delta^\dag \after \delta = \id$,  and $\swapmor
\after \delta = \delta$, where $\swapmor \colon X
\tensor X \to X \tensor X$ is the swap isomorphism~\cite{coeckepaquette:practisingphysicist, 
 coeckepavlovicvicary:bases, abramskyheunen:hstar}. We depict $\delta$ as
$\vcenter{\hbox{\begin{tikzpicture}[scale=0.15]
     \node [style=none] (2) at (-1, 1) {};
     \node [style=none] (3) at (1, 1) {};
     \node [style=dot] (4) at (0, 0) {};
     \node [style=none] (5) at (0, -1) {};
     \draw[looseness=1.25, bend left=315] (2.center) to (4);
     \draw (4) to (5.center);
     \draw[looseness=1.25, bend right=45] (4) to (3.center);
   \end{tikzpicture}}}$.
In the category $\Cat{fdHilb}$ of finite-dimensional Hilbert spaces,
classical structures correspond to orthonormal
bases~\cite{coeckepavlovicvicary:bases, abramskyheunen:hstar}: an
orthonormal basis $\{\ket{i}\}_i$ induces a classical structure by
$\ket{i} \mapsto \ket{ii}$, and the
basis is retrieved from a classical structure as the set of nonzero
\emph{copyables} $\{ x \in X \mid \delta(x)=x \tensor x \} \backslash
\{0\}$. 
This legitimizes using classical structures to model (the
copying of) classical data. 

The other main ingredient we work with is the \emph{CPM-construction},
that we now briefly recall, turning a dagger compact closed 
category $\cat{C}$ into another one by taking the completely positive
morphisms~\cite{selinger:completelypositive}. The latter is given by lifting the Stinespring characterization of completely positive maps to a definition:
\begin{itemize}
  \item The objects of $\CPM(\cat{C})$ are those of $\cat{C}$. To
    distinguish the category they live in, we will denote
    objects of $\cat{C}$ by $X$, and objects of $\CPM(\cat{C})$ by $\cp{X}$.
  \item The morphisms $\cp{X} \to \cp{Y}$ in $\CPM(\cat{C})$ are the
    morphisms of $\cat{C}$ of the form $(\id \tensor \varepsilon \tensor \id)
    \after (f_* \tensor f)$ for some $f \colon X \to Z \tensor
    Y$ in $\cat{C}$ called \emph{Kraus morphism}. In
    graphical~\cite{selinger:graphicallanguages} terms: 
    \begin{align*}\label{eq:cp_morphism}
      \cp{f} \;\;=\;\; \left(\vcenter{\hbox{\begin{tikzpicture}
	\begin{pgfonlayer}{nodelayer}
          \node [style=none] (0) at (-1, 1) {$Y^*$};
          \node [style=none] (1) at (1, 1) {$Y$};
          \node [style=none] (2) at (-1, 0) {};
          \node [style=none] (3) at (-0.5, 0) {};
          \node [style=none] (4) at (0.5, 0) {};
          \node [style=none] (5) at (1, 0) {};
          \node [style=none] (6) at (-1, -1) {$X^*$};
          \node [style=none] (7) at (1, -1) {$X$};
          \node [style=none] at (0, 0.8) {$Z$};
          \node [style=box] at (-.75,0) {$\;\;f_*\;\;$};
          \node [style=box] at (.75,0) {$\,\;\;f\;\;\,$};
	\end{pgfonlayer}
	\begin{pgfonlayer}{edgelayer}
          \draw[arrow] (0) to (2);
          \draw[arrow, looseness=1.75, bend right=90] (4) to (3);
          \draw[arrow] (2) to (6);
          \draw[arrow] (7) to (5);
          \draw[arrow] (5) to (1);
	\end{pgfonlayer}
    \end{tikzpicture}}}\right)
    \end{align*}
    Here, $\varepsilon \colon Y \otimes Y^* \to I$ and $f_* \colon Y^*
    \to X^*$ are the counit and conjugation of the compact structure.
    To distinguish their home category, we denote
    morphisms in $\cat{C}$ as $f$, and those in $\CPM(\cat{C})$ as $\cp{f}$.
  \item Composition and dagger are as in $\cat{C}$.
  \item The tensor product of objects is as in $\cat{C}$.
  \item The tensor product of morphisms is given by
    $
        (\cp{f} \cp{\tensor} \cp{g}) 
      = ((123)(4)) \after (f_* \tensor f \tensor g_* \tensor g) \after ((132)(4))
    $,
    where the group theoretic notation $((132)(4))$ denotes the
    canonical permutation built out of identities and swap
    morphisms using tensor product and composition. Graphically, this
    looks as follows.
    \[
      \left(
      \vcenter{\hbox{\begin{tikzpicture}
        \begin{pgfonlayer}{nodelayer}
          \node [style=none] (0) at (-1, 1) {$B^*$};
          \node [style=none] (1) at (1, 1) {$B$};
          \node [style=none] (2) at (-1, 0) {};
          \node [style=none] (3) at (-0.5, 0) {};
          \node [style=none] (4) at (0.5, 0) {};
          \node [style=none] (5) at (1, 0) {};
          \node [style=none] (6) at (-1, -1) {$A^*$};
          \node [style=none] (7) at (1, -1) {$A$};
          \node [style=none] at (0, 0.75) {};
          \node [style=box] at (-.75,0) {$\;\;f_*\;\;$};
          \node [style=box] at (.75,0) {$\,\;\;f\;\;\,$};
        \end{pgfonlayer}
        \begin{pgfonlayer}{edgelayer}
          \draw[arrow] (0) to (2);
          \draw[arrow, looseness=1.75, bend right=90] (4) to (3);
          \draw[arrow] (2) to (6);
          \draw[arrow] (7) to (5);
          \draw[arrow] (5) to (1);
        \end{pgfonlayer}
      \end{tikzpicture}}}
      \right) \tensor \left(
      \vcenter{\hbox{\begin{tikzpicture}
        \begin{pgfonlayer}{nodelayer}
          \node [style=none] (0) at (-1, 1) {$D^*$};
          \node [style=none] (1) at (1, 1) {$D$};
          \node [style=none] (2) at (-1, 0) {};
          \node [style=none] (3) at (-0.5, 0) {};
          \node [style=none] (4) at (0.5, 0) {};
          \node [style=none] (5) at (1, 0) {};
          \node [style=none] (6) at (-1, -1) {$C^*$};
          \node [style=none] (7) at (1, -1) {$C$};
          \node [style=none] at (0, 0.75) {};
          \node [style=box] at (-.75,0) {$\;\;g_*\vphantom{f_*}\;\;$};
          \node [style=box] at (.75,0) {$\,\;\;g\vphantom{f}\;\;\,$};
        \end{pgfonlayer}
        \begin{pgfonlayer}{edgelayer}
          \draw[arrow] (0) to (2);
          \draw[arrow, looseness=1.75, bend right=90] (4) to (3);
          \draw[arrow] (2) to (6);
          \draw[arrow] (7) to (5);
          \draw[arrow] (5) to (1);
        \end{pgfonlayer}
      \end{tikzpicture}}}
      \right)  \; = \; \left(
      \vcenter{\hbox{\begin{tikzpicture}[scale=0.66]
        \begin{pgfonlayer}{nodelayer}
          \node [style=None] (0) at (-2.5, 2) {$D^*$};
          \node [style=None] (1) at (-1, 2) {$B^*$};
          \node [style=None] (2) at (1, 2) {$B$};
          \node [style=None] (3) at (2.5, 2) {$D$};
          \node [style=none] (7) at (-2.5, -0) {};
          \node [style=none] (8) at (-1.75, -0) {};
          \node [style=none] (9) at (-1, -0) {};
          \node [style=none] (10) at (-0.5, -0) {};
          \node [style=none] (11) at (0.5, -0) {};
          \node [style=none] (12) at (1, -0) {};
          \node [style=none] (13) at (1.75, -0) {};
          \node [style=none] (14) at (2.5, -0) {};
          \node [style=None] (18) at (-2.5, -2) {$C^*$};
          \node [style=None] (19) at (-1, -2) {$A^*$};
          \node [style=None] (20) at (1, -2) {$A$};
          \node [style=None] (21) at (2.5, -2) {$C$};
          \node [style=box] at (-2.25,0) {$\;\;g_*\vphantom{f_*}\;\;$};
          \node [style=box] at (-.75,0) {$\,\;\;f_*\;\;\,$};
          \node [style=box] at (.75,0) {$\;\;f\;\;$};
          \node [style=box] at (2.25,0) {$\,\;\;g\vphantom{f}\;\;\,$};
        \end{pgfonlayer}
        \begin{pgfonlayer}{edgelayer}
          \draw[arrow] (21.center) to (14.center);
          \draw[arrow] (0.center) to (7.center);
          \draw[arrow, markat=.35] (1.center) to (9.center);
          \draw[arrow, markat=.7] (12.center) to (2.center);
          \draw[arrow,looseness=1.75, bend right=90] 
               (11.center) to (10.center);
          \draw[arrow] (7.center) to (18.center);
          \draw[arrow] (9.center) to (19.center);
          \draw[arrow] (20.center) to (12.center);
          \draw[arrow] (14.center) to (3.center);
          \draw[arrow, cross, looseness=1.25, bend right=90] (13) to (8);
         \end{pgfonlayer}
      \end{tikzpicture}}}
      \right)
    \]
  \item The swap isomorphism $\cp{\swapmor} \colon \cp{X \tensor Y}
    \to \cp{Y \tensor X}$ in $\CPM(\cat{C})$ is given by $\swapmor
    \tensor \swapmor$. 
    \[\begin{tikzpicture}[scale=0.66]
        \begin{pgfonlayer}{nodelayer}
          \node [style=none] (0) at (-1.5, 1) {$X^*$};
          \node [style=none] (1) at (-0.5, 1) {$Y^*$};
          \node [style=none] (2) at (0.5, 1) {$Y$};
          \node [style=none] (3) at (1.5, 1) {$X$};
          \node [style=none] (4) at (-1.5, -1) {$Y^*$};
          \node [style=none] (5) at (-0.5, -1) {$X^*$};
          \node [style=none] (6) at (0.5, -1) {$X$};
          \node [style=none] (7) at (1.5, -1) {$Y$};
        \end{pgfonlayer}
        \begin{pgfonlayer}{edgelayer}
          \draw[arrow, markat=.3,  markwith=<,looseness=1.25,
                out=270, in=90] (2) to (7); 
          \draw[arrow, markat=.3, looseness=1.25, out=270, in=90] (1) to (4);  
          \draw[arrow, markat=.75,
                cross, looseness=1.25, out=90, in=270] (6) to (3); 
          \draw[arrow, markat=.3, 
                cross, looseness=1.25, out=270, in=90] (0) to (5); 
        \end{pgfonlayer}
    \end{tikzpicture}\]
\end{itemize}

There are two functors of particular interest between $\cat{C}$ and
$\CPM(\cat{C})$. They form the components of a natural transformation
when varying the base category
$\cat{C}$~\cite{mccurdyselinger:basicdaggercompactcategories}. 

\begin{lemma}
\label{lem:Ffunctor}
  Consider the functors 
  \begin{align*}
    F & \colon \cat{C} \to \CPM(\cat{C})
    & F(X) & = \cp{X}
    & F(f) & = f_* \tensor f, \\
    G & \colon \CPM(\cat{C}) \to \cat{C}
    & G(\cp{X}) & = X^* \tensor X
    & G(\cp{f}) & = (\id \tensor \varepsilon \tensor \id) \after (f_* \tensor f).
 \end{align*}
  Then $F$ is a symmetric (strict) monoidal functor that preserves dagger, \\
  and $G$ is a symmetric (strong) monoidal functor that preserves dagger. 
\end{lemma}
\begin{proof}
  The first statement is easily verified by taking identities for the required
  morphisms $\cp{I} \to F(I)$ and $F(X) \tensor F(Y) \to F(X \tensor
  Y)$. As to the second statement, $G$ preserves daggers almost by definition.
  To verify that $G$ is (strong) monoidal, one can take the
  required morphisms $\varphi \colon I \to G(\cp{I}) = I^* \tensor I$
  and $\varphi_{X,Y} \colon G(\cp{X}) \tensor G(\cp{Y}) = X^* \tensor
  X \tensor Y^* \tensor Y \to Y^* \tensor X^* \tensor X \tensor Y
  \cong G(\cp{X \tensor Y})$ to be the canonical isomorphisms of that
  type. In particular, $\varphi_{X,Y} = (123)(4)$. 
  This makes the required diagrams commute. For example, the `left
  unit' condition boils down to $(\rho_X)_* = \lambda_{X^*} \after (i
  \otimes \id[X^*])$, where $i$ is the canonical isomorphism $I^* \to
  I$. This follows from unitarity of $\rho_X$, because
  $(\rho_X)_* = (\rho_X)^{\dag *} = (\rho_X^{-1})^* = (\rho_X^*)^{-1}$,
  and the diagram
  \[\xymatrix{
      (X \tensor I)^* \ar_-{\cong}[d] 
    & X^* \ar_-{\rho_X^*}[l] \\
      I^* \tensor X^* \ar_-{i \tensor \id}[r] 
    & I \tensor X^* \ar_-{\lambda_{X^*}}[u]
  }\]
  commutes by the coherence theorem for compact closed categories.
  Finally, 
  \[
    \varphi \after \swapmor
    = \vcenter{\hbox{\begin{tikzpicture}[scale=0.75]
	\begin{pgfonlayer}{nodelayer}
          \node [style=None] (0) at (-1.5, 1) {$X^*$};
          \node [style=None] (1) at (-0.5, 1) {$Y^*$};
          \node [style=None] (2) at (0.5, 1) {$Y$};
          \node [style=None] (3) at (1.5, 1) {$X$};
          \node [style=none] (4) at (-1.5, -0) {};
          \node [style=none] (5) at (-0.5, -0) {};
          \node [style=none] (6) at (0.5, -0) {};
          \node [style=none] (7) at (1.5, -0) {};
          \node [style=None] (8) at (-1.5, -1) {$X^*$};
          \node [style=None] (9) at (-0.5, -1) {$X$};
          \node [style=None] (10) at (0.5, -1) {$Y^*$};
          \node [style=None] (11) at (1.5, -1) {$Y$};
	\end{pgfonlayer}
	\begin{pgfonlayer}{edgelayer}
          \draw[looseness=1.50, out=-90, in=90, -<] (2.center) to (5.center);
          \draw[looseness=1.50, out=-90, in=90, ->] (1.center) to (4.center);
          \draw[looseness=0.75, out=270, in=90] (5.center) to (11.center);
          \draw[looseness=0.75, out=-90, in=90] (4.north) to (10.center);
          \draw[cross, -<] (3.center) to (7.center);
          \draw[cross, looseness=0.75, out=270, in=90] 
               (7.center) to (9.center);
 	  \draw[cross, looseness=0.75, out=270, in=90] 
               (6.north) to (8.center);
          \draw[cross, out=270, in=90, ->] (0.center) to (6.south);
	\end{pgfonlayer}
      \end{tikzpicture}}}
    = \vcenter{\hbox{\begin{tikzpicture}[scale=0.75]
	\begin{pgfonlayer}{nodelayer}
          \node [style=None] (0) at (-1.5, 1) {$X^*$};
          \node [style=None] (1) at (-0.5, 1) {$Y^*$};
          \node [style=None] (2) at (0.5, 1) {$Y$};
          \node [style=None] (3) at (1.5, 1) {$X$};
          \node [style=none] (4) at (-1.5, -0) {};
          \node [style=none] (5) at (-0.5, -0) {};
          \node [style=none] (6) at (0.5, -0) {};
          \node [style=none] (7) at (1.5, -0) {};
          \node [style=None] (8) at (-1.5, -1) {$X^*$};
          \node [style=None] (9) at (-0.5, -1) {$X$};
          \node [style=None] (10) at (0.5, -1) {$Y^*$};
          \node [style=None] (11) at (1.5, -1) {$Y$};
	\end{pgfonlayer}
	\begin{pgfonlayer}{edgelayer}
          \draw[out=-90, in=90] (4.north) to (10.center);
          \draw[out=270, in=90, -<] (2.center) to (7.center);
          \draw (7.center) to (11.center);
          \draw[in=90, out=270, ->] (1.center) to (4.south);
          \draw[cross, out=270, in=90] (6.center) to (9.center);
          \draw[cross, looseness=1.25, in=90, out=270, -<] 
               (3.center) to (6.center);
          \draw[cross, looseness=1.25, out=90, in=270, -<] 
               (8.center) to (5.center);
	  \draw[cross, looseness=1.25, out=90, in=270] 
               (5.center) to (0.center);
	\end{pgfonlayer}
      \end{tikzpicture}}}
    = \cp{\swapmor} \after \varphi,
 \]
  so $G$ is in fact symmetric strong monoidal.
\end{proof}

\section{Completely positive classical structures}\label{sec:cpcs}

This section considers classical structures in categories
$\CPM(\cat{C})$ for various dagger compact categories $\cat{C}$,
starting from the following corollary of Lemma~\ref{lem:Ffunctor}.

\begin{corollary}
\label{cor:trivialcs}
  Every classical structure on $X$ in $\cat{C}$ induces a classical
  structure on $\cp{X}$ in $\CPM(\cat{C})$.  \\
  Every classical structure on $\cp{X}$ in $\CPM(\cat{C})$ induces a
  classical structure on $X^* \tensor X$ in $\cat{C}$.   
  \qed
\end{corollary}

Classical structures in $\CPM(\cat{C})$ induced by classical structures
in $\cat{C}$ are called \emph{canonical}.
We will prove that for $\cat{C}=\Cat{Rel}$, $\cat{C}=\Cat{Cob}$, and
for $\cat{C}$ the free dagger compact category on one generator,
any classical structure in $\CPM(\cat{C})$ is canonical. We conjecture
the same is the case for $\cat{C}=\Cat{fdHilb}$ and provide strong
evidence for this conjecture.

\paragraph{Sets and relations}

Let us start with the category $\Cat{Rel}$ of sets and relations as a
toy example. In general, a relation $R \subseteq (X \times X) \times
(Y \times Y)$ is completely positive when
\begin{align}
  (x',x)R(y',y) & \Longleftrightarrow (x,x')R(y,y'), \label{eq:cprel1} \\
  (x',x)R(y',y) & \Longrightarrow (x,x)R(y,y). \label{eq:cprel2}
\end{align}
By~\cite{pavlovic:frobeniusinrel}, a classical structure on an object $X$ in
$\CPM(\Cat{Rel})$ is a direct sum of Abelian groups, the union of whose
carrier sets is $X \times X$. Additionally, the group 
multiplication must be a completely positive relation in the above
sense. This means the following, writing the group additively:
\begin{align}
  (a,b) + (c,d) = (e,f) & \Longleftrightarrow (c,d) + (a,b) = (f,e), 
  \label{eq:cpab1} \\
  (a,b) + (c,d) = (e,f) & \Longrightarrow (c,d) + (c,d) = (f,f). 
  \label{eq:cpab2}
\end{align}
In particular, if we take $(a,b)=0$, then \eqref{eq:cpab1} implies
that $(c,d) = (c,d)+(a,b)=(d,c)$. Hence it is already forced that
$c=d$. That is, the classical structure on $X$ in $\CPM(\Cat{Rel})$
must be canonical. This proves the following theorem.

\begin{theorem}
  Any classical structure in $\CPM(\Cat{Rel})$ is canonical. 
  \qed
\end{theorem}

\paragraph{Cobordisms}

Now consider the category $\Cat{Cob}$ of 1d compact closed manifolds and
2d cobordisms between them is a dagger compact category. It
is the free dagger compact category on a classical
structure~\cite{kock:frobenius}, and is very interesting from the
point of view of quantum field theory~\cite{baez:quandaries}. However,
it has the property that every isometry is automatically unitary: if
we are to have a homeomorphism 
\[
  \vcenter{\hbox{\includegraphics[height=10ex]{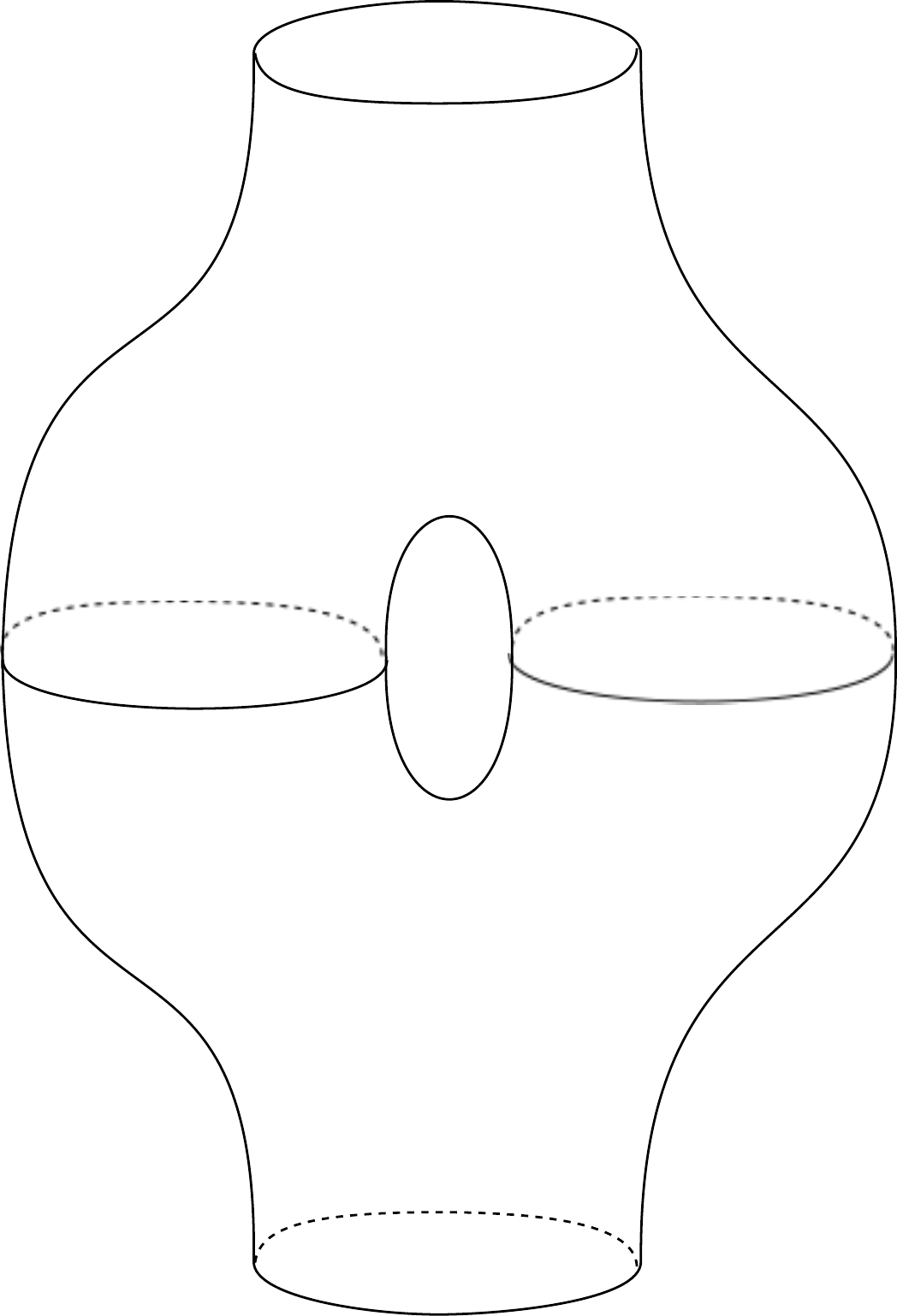}}} 
  \quad \stackrel{\cong}{\longrightarrow} 
  \vcenter{\hbox{\includegraphics[height=10ex]{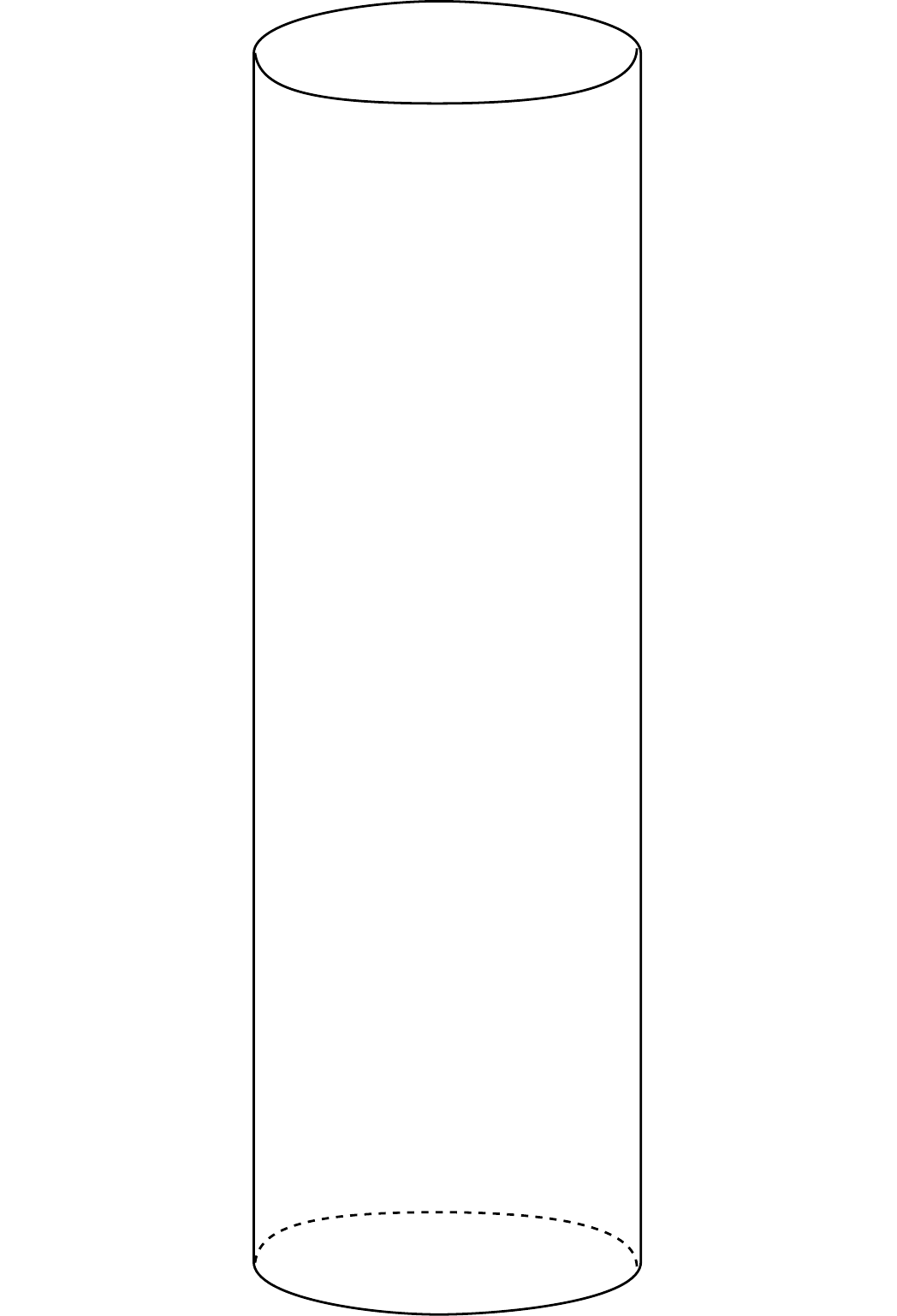}}}
\]
the left-hand cobordism cannot actually have nonzero
genus. Therefore $\Cat{Cob}$ has no nontrivial classical structures at
all, as does $\CPM(\Cat{Cob})$. This makes the following proposition
trivial, but nevertheless true.

\begin{proposition}
  Any classical structure in $\CPM(\Cat{Cob})$ is canonical.
  \qed
\end{proposition}

\paragraph{Free dagger compact category}

Let us now have a look at the free dagger compact category $\cat{C}$ on one
generator. It is described explicitly
in~\cite{abramsky:freecompactcategory}, whose notation we adopt here.
Concretely, the objects of $\cat{C}$ are $(m,n) \in
\mathbb{N}^2$. Morphisms $(m,n) \to (p,q)$ exist only when $m+q=n+p$,
and are pairs $(s,\pi)$ of a natural number $s$ and a permutation $\pi
\in S(m+q)$. The identity is $(0,\id)$, and composition is given by the
so-called execution formula as $(t,\sigma) \after (s, \pi) = (s+t,
\mbox{ex}(\sigma,\pi))$, see~\cite{abramsky:freecompactcategory}. The
monoidal structure is $(m,n) \tensor (p,q) = (m+p,n+q)$ on objects and
$(s,\pi) \tensor (t,\sigma)=(s+t,\pi + \sigma)$ on morphisms. Finally,
the dagger is given by $(s,\pi)^\dag = (s,\pi^{-1})$. It turns out
that the situation is similar to that with cobordisms.

\begin{proposition}
  Let $\cat{C}$ be the free dagger compact category on one generator.
  Any classical structure in $\CPM(\cat{C})$ is canonical.
\end{proposition}
\begin{proof}
  We show that in both categories $\cat{C}$ and $\CPM(\cat{C})$,
  the only morphisms $\delta \colon X \to X \tensor X$ satisfying
  $\delta^\dag \after \delta = \id$ are the trivial ones with $X=I$.
  First of all, if $\delta=(s,\pi) \colon (m,n) \to (m,n) \tensor (m,n) = (2m,2n)$
  is a morphism in $\cat{C}$, then $m=n$, and hence $\pi \in
  S(3n)$. Now $\delta^\dag \after \delta = (s,\pi^{-1}) \after
  (s,\pi) = (2s, \mbox{ex}(\pi^{-1},\pi))$. Writing both $\pi$ and
  $\mbox{ex}(\pi^{-1},\pi)$ as matrices, as
  in~\cite{abramsky:freecompactcategory}, we find
  $\mbox{ex}(\pi^{-1},\pi)_{12} = \pi_{11} \circ \pi^{-1}_{11} = 1$. But
  $\id[12] = 0$. So $\delta^\dag \after \delta = \id$ only when
  $m=n=0$. That is, $\cat{C}$ only has trivial classical structures.

  Similarly, working out the upper-right entry of the matrix of
  $\cp{\delta}^\dag \after \cp{\delta}$ for a
  morphism $\cp{\delta}=\cp{(s,\pi)}$ in $\CPM(\cat{C})$, we find that it is the
  union of $\pi_{11} \circ \pi^{-1}_{11}$ with some other terms. In
  particular, it contains $\id$. Hence $\cp{\delta}^\dag \after
  \cp{\delta} = \textbf{id}$ only if $\mbox{dom}(\cp{\delta}) =
  \cp{I}$, so $\CPM(\cat{C})$ only has trivial classical structures, too.
\end{proof}

\paragraph{Finite-dimensional Hilbert spaces}

The rest of this section concentrates on the category $\CPM(\Cat{fdHilb})$.
Here, we can identify the object $\cp{X}$ with $\field{C}^n$, and $X^* \tensor X$ 
with the Hilbert space $M_n(\field{C})$ of $n$-by-$n$ matrices under
the Hilbert-Schmidt inner product $\inprod{\rho}{\sigma}=\Tr(\rho^\dag
\sigma)$. One strategy to investigate classical structures in this category, inspired by the strategy that worked so well for $\Cat{Rel}$, could be as follows: 
\begin{enumerate}
  \item Start with a classical structure $\cp{\delta}$ on $\field{C}^n$ in
    $\CPM(\Cat{fdHilb})$.
  \item Apply Corollary~\ref{cor:trivialcs} to obtain a classical
    structure $\delta = G(\cp{\delta})$ on $M_n(\field{C})$.
 \item Apply~\cite{coeckepavlovicvicary:bases} to obtain an orthonormal basis
    $\{\alpha\}$ for $M_n(\field{C})$.
  \item Investigate when $\delta \colon \ket{\alpha} \mapsto
    \ket{\alpha\alpha}$ is a completely positive map.
\end{enumerate}
Unfortunately, the last step proves to be quite difficult. A first
thought might be to use Choi's theorem, resulting in the following characterization.

\begin{theorem}
\label{thm:choi}
  Classical structures $\cp{\delta}$ on $\field{C}^n$ in
  $\CPM(\Cat{fdHilb})$ correspond to orthonormal bases $\alpha$ of
  $M_n(\field{C})$ for which 
  $
      \Delta_{j''k''j'k'jk} 
    = \sum_\alpha \overline{\alpha_{jk}} \alpha_{j'k'} \alpha_{j''k''}
  $
  are matrix entries of a positive semidefinite $\Delta \colon
  (\field{C}^n)^{\tensor 3} \to   (\field{C}^n)^{\tensor 3}$.
\end{theorem}
\begin{proof}
  Setting $\ket \psi = \sum_i \ket{ii}$ shows that  
  \[
    \delta(\rho) 
   = \sum_\alpha \bra{\psi} (\alpha^\dag \tensor \id) \after
        (\rho \tensor \id) \ket{\psi} (\alpha \tensor \alpha) \\
   = \sum_\alpha \mbox{Tr}(\alpha^\dag  \rho) \cdot (\alpha
        \tensor \alpha), 
  \]
  so that in particular $\delta(\ket{i}\bra{j}) = \sum_\alpha
  \overline{\alpha_{ij}} \cdot (\alpha \tensor \alpha)$. 
  Recall that Choi's theorem states that a morphism $\delta \colon X^*
  \tensor X \to Y^* \tensor Y$ in $\Cat{fdHilb}$ is completely
  positive if and only if the map $\Delta \colon X \tensor Y \to X \tensor Y$
  given by $\bra{aj}\Delta\ket{bk} =
  \bra{a}\delta(\ket{j}\bra{k})\ket{b}$ for $a,b \in Y$ and $j,k \in 
  X$ is positive semidefinite~\cite[3.14]{paulsen:completelypositive}. 
  Applying Choi's theorem to the map $\delta$ defined 
  by $\alpha \mapsto \alpha \tensor \alpha$ and taking $a =
  \ket{j''j'}$ and $b=\ket{k''k'}$ yields
  \[
        \Delta_{j''k''j'k'jk}
    = \bra{j''j'} \delta(\ket{j}\bra{k}) \ket{k''k'} 
    = \sum_\alpha \overline{\alpha_{jk}} \bra{j''j'} \alpha \tensor
        \alpha \ket{k''k'} 
    = \sum_\alpha \overline{\alpha_{jk}} \alpha_{j'k'}
    \alpha_{j''k''}.
  \qedhere
  \]
\end{proof}

The condition of the previous theorem is not vacuous. For
example, the normalized Pauli matrices
\[
  \alpha_1=\frac{1}{\sqrt{2}}\begin{pmatrix} 1 & 0 \\ 0 & 1 \end{pmatrix}, \;
  \alpha_2=\frac{1}{\sqrt{2}}\begin{pmatrix} 0 & 1 \\ 1 & 0 \end{pmatrix}, \;
  \alpha_3=\frac{1}{\sqrt{2}}\begin{pmatrix} 0 & -i \\ i & 0 \end{pmatrix}, \;
  \alpha_4=\frac{1}{\sqrt{2}}\begin{pmatrix} 1 & 0 \\ 0 & -1 \end{pmatrix}.
\]
form an orthonormal basis for $M_2(\field{C})$. But their Choi matrix
$\Delta$ is not even positive: 
  \[
    \sqrt{2} \alpha_1 + \alpha_2 + \alpha_4 =
    \begin{pmatrix}1+\frac{1}{\sqrt{2}} & \frac{1}{\sqrt{2}} \\
    \frac{1}{\sqrt{2}} & 1-\frac{1}{\sqrt{2}}\end{pmatrix} \geq 0,
  \]
  as its eigenvalues are $\{2,0\}$, but its image under $\delta$ is
  \[
  \sqrt{2} \alpha_1 \tensor \alpha_1 + \alpha_2 \tensor \alpha_2 +
  \alpha_4 \tensor \alpha_4 = \begin{pmatrix} \frac{1}{\sqrt{2}}+\frac{1}{2}
    & 0 & 0 & \frac{1}{2} \\ 0 & \frac{1}{\sqrt{2}}-\frac{1}{2} & \frac{1}{2} &
    0  \\ 0 & \frac{1}{2} & \frac{1}{\sqrt{2}} - \frac{1}{2} & 0 \\
    \frac{1}{2} & 0 & 0 & \frac{1}{\sqrt{2}} + \frac{1}{2}\end{pmatrix},
 \]
 which has an eigenvalue $\frac{1}{2}(\sqrt{2}-2)<0$ and therefore is
 not positive semidefinite.

\begin{evidence}
  The condition of Theorem~\ref{thm:choi} is
  easy to verify with computer algebra software. We can also
  generate random orthonormal bases when $X=\field{C}^n$, \ie when
  $X^* \tensor X$ is the C*-algebra $M_n(\field{C})$ of
  $n$-by-$n$-matrices: generate $n^2$ random matrices; with large
  probability they are linearly independent; use the Gram-Schmidt
  procedure to turn them into an orthonormal basis of $M_n$. Running
  this test for 1000 randomly chosen orthonormal bases did not
  provide a counterexample to our conjecture for $n=2,3$.
\end{evidence}

The previous Theorem and Numerical evidence neither provide much intuition
nor have direct physical significance. We end this section by providing some
more natural sufficient conditions, and by listing a large number of
necessary and sufficient conditions equivalent to our conjecture. This
partial progress is enough to enable novel physical
applications outlined in the next section.

\begin{proposition}
  \label{prop:sufficientconditions}
  Let $\{\alpha_1,\ldots,\alpha_N\}$ be an orthonormal basis of $X^* \tensor X$, and let
  $\delta$ be the induced classical structure in $\Cat{fdHilb}$. 
  If $\delta$ is completely positive, then:
  \begin{itemize}
  \item[(a)] the set of copyables is closed under adjoints: $\{\alpha_1,\ldots,\alpha_N\} =
    \{\alpha_1^\dag,\ldots,\alpha_N^\dag\};$
  \item[(b)] copyables are closed under absolute values:
    $\{|\alpha_1|,\ldots,|\alpha_N|\} \subseteq \{0,\alpha_1,\ldots,\alpha_N\};$
  \item[(c)] the map $\delta$ satisfies
    $\delta(\id)\delta(ab)=\delta(a)\delta(b)$ for all $a,b \colon X
    \to X$;
  \item[(d)] positive semidefinite copyables commute pairwise. 
\end{itemize}
\end{proposition}
\begin{proof}
  To see (a), notice that if $\delta$ is positive, then it preserves (names of)
  adjoints~\cite[Lemma~2.3.1]{bhatia:positivematrices}, so that in
  particular $\delta(\alpha_j^\dag) = \delta(\alpha_j)^\dag =
  \alpha_j^\dag \tensor \alpha_j^\dag$. Therefore $\alpha_j^\dag$ is
  (also) one of the basis vectors. 

  For (b) and (c) we first prove the auxiliary result that $ab=0$ if and
  only if $\inprod{a}{b}=0$ for positive definite matrices $a,b$. Say
  $a=x^\dag x$ and $b=y^\dag y$. If $ab=0$, then
  $0=\Tr(x^\dag xy^\dag y)=\inprod{a}{b}$. Conversely, if
  $0=\inprod{a}{b}=\Tr(x^\dag xy^\dag y)=\Tr(yx^\dag xy^\dag)=\|xy^\dag\|^2$ then
  $xy^\dag$=0 and hence $ab=0$. 

  By isometry of $\delta$ we can therefore conclude that
  $\delta(a)\delta(b)=0$ whenever $ab=0$ for positive $a,b$. Hence (c)
  follows from~\cite[Theorem~2]{gardner:absolutevalue}, as well as the
  fact that $\delta$ preserves absolute values. Hence if $\alpha$ is copyable, then
  $\delta(|\alpha|)=|\delta(\alpha)|=|\alpha \tensor \alpha|=|\alpha|
  \tensor |\alpha|$, and so $|\alpha|$ is copyable, too, establishing (b).

  Finally, (d) follows from a generalized version of the no-cloning
  theorem~\cite[Corollary~3]{barnumetal:nobroadcasting}.
\end{proof}

\begin{lemma}\label{lem:deltadagid}
  If $\delta$ is completely positive, then $\delta^\dag(\id)$: 
  \begin{enumerate}
  \item[(a)] is positive semidefinite;
  \item[(b)] is invertible, and hence positive definite;
  \item[(c)] satisfies $\delta^\dag(\id) \geq \id$, and hence its
    eigenvalues are at least 1.
  \end{enumerate}
\end{lemma}
\begin{proof}
  To see (a), notice that $\id$ is a positive definite matrix and that
  $\delta^\dag$ is a positive map. For (b),
  apply~\cite[Lemma~2.2]{choieffros:injectivity} to get a completely
  positive map $d$ such that $d^\dag(\id)=\id$ and $\delta^\dag =
  (\sqrt{f} \tensor \sqrt{f}) \after d$, where $f=\delta^\dag(\id)$, and observe that
  \[
  \dim(X)^2 = \rank(\delta^\dag) = \rank((\sqrt{f} \tensor \sqrt{f})
  \after d^\dag) \leq \rank(\sqrt{f} \tensor \sqrt{f}) \leq \dim(X)^2,
  \]
  so that $\rank(\sqrt{f}) = \dim(X)$, whence $\sqrt{f}$ is
  invertible. Therefore also $\delta^\dag(\id)$ is invertible.

  Notice that $\delta(\id)^\dag=\delta(\id^\dag)=\delta(\id)$ since
  $\delta$ is positive by Proposition~\ref{prop:sufficientconditions}(a).
 The map $\frac{1}{\opnorm{\delta(\id)}}\delta(\id)$ is a contraction, and
  so~\cite[Proposition~1.3.1]{bhatia:positivematrices} the matrix
  \[
    \frac{1}{\opnorm{\delta(\id)}}
    \begin{pmatrix} \id & \delta(\id) \\ \delta(\id) & \id \end{pmatrix}
  \]
  is positive semidefinite. Because $\delta^\dag$ is completely positive and
  $\delta^\dag \after \delta(\id)=\id$, therefore also the matrix
  \[
    \frac{1}{\opnorm{\delta(\id)}}
    \begin{pmatrix} \delta^\dag(\id) & \id \\ 
                    \id & \delta^\dag(\id) \end{pmatrix} 
  \]
  is positive semidefinite. It follows from (b)
  and~\cite[Theorem~1.3.3]{bhatia:positivematrices} that
  $$\frac{1}{\opnorm{\delta(\id)}} \delta^\dag(\id) \geq
  \frac{1}{\opnorm{\delta(\id)}} \left(\frac{1}{\opnorm{\delta(\id)}}
  \delta^\dag(\id)\right)^{-1} \frac{1}{\opnorm{\delta(\id)}},$$ and therefore
  $\delta^\dag(\id) \geq (\delta^\dag(\id))^{-1}$. Finally, since
  $\delta^\dag(\id)$ is diagonalizable by (a), this yields
  $\delta^\dag(\id) \geq \id$, establishing (c).
\end{proof}

\begin{theorem}
\label{thm:equivalentconditions}
  The following are equivalent:
  \begin{enumerate}
    \item[(a)] A classical structure $\cp{\delta}$ on $\cp{X}$ in
      $\CPM(\Cat{fdHilb})$ is canonical;
    \item[(b)] $\delta$ preserves pure states;
    \item[(c)] $\delta$ is a trace-preserving map;
    \item[(d)] $\delta^\dag$ is unital: $\delta^\dag(\id[X \tensor X])
      = \id[X]$;
    \item[(e)] $\delta^\dag(\id[X \tensor X]) \leq \id[X]$;
    \item[(f)] $\delta$ is a trace-nonincreasing map;
    \item[(g)] $\delta^\dag(\id[X \tensor X])$ has operator norm at
     most 1;\footnote{One might think that
      isometry of $\delta$ means that it has operator norm 1, and
      hence so does $\delta^\dag$, and hence (by applying the Russo-Dye
      theorem~\cite[Corollary~2.3.8]{bhatia:positivematrices}) so does
      $\delta^\dag(\id)$. However, the first operator
      norm is taken with respect to the Hilbert-Schmidt norm on
      $M_n(\field{C})$, and the last operator norm is taken with
      respect to the operator norm on $M_n(\field{C})$. Hence all
      one can conclude from this reasoning is that $\dim(X)^{-1/2}
      \leq \opnorm{\delta^\dag(\id)} \leq \dim(X)^{1/2}$.}
    \item[(h)] $\id[X]=\sum A$ for some subset $A$ of the copyables;
    \item[(i)] $\Tr(\alpha)=1$ for positive semidefinite copyables $\alpha$;
    \item[(j)] $\rank(\alpha)=1$ for nonzero copyables $\alpha$;
    \item[(k)] if $\alpha$ is copyable, then so is $\alpha^\dag \alpha$;
    \item[(l)] the set of copyables is closed under composition:
      $\alpha_i \alpha_j \in \{0,\alpha_1,\ldots,\alpha_N\}$;
    \item[(m)] $\delta(\id[X])$ is idempotent;
    \item[(n)] the ancilla in a minimal Kraus decomposition of
      $\delta$ is one-dimensional.
\end{enumerate}
\end{theorem}
\begin{proof}
\begin{description}
  \item[(b$\Rightarrow$a)] If $\cp{\delta}$ sends pure states $\ket{\psi}
    \bra{\psi}$ to pure states $\ket{\varphi}\bra{\varphi}$, it defines a
    linear map $g$ sending $\ket{\psi}$ to the unique eigenvector
    $\ket{\varphi}$ of $\cp{\delta}(\ket{\psi}\bra{\psi})$, so that
    $\cp{\delta}=g_* \tensor g$. One then readily verifies that the
    axioms for classical structures on $\cp{\delta}$ imply that $g$ is a
    classical structure in $\Cat{fdHilb}$. 
  \item[(c$\Rightarrow$b)] Since completely positive trace-preserving
    maps send quantum states to quantum states (with eigenvalues
    bounded above by 1), isometry of $\delta$ yields 
   \[
       \Tr \left(\cp{\delta}(\ket{\psi}\bra{\psi})^2\right)
     = \Tr((\ket{\psi}\bra{\psi})^2) 
     = 1,
   \]
   and therefore $\cp{\delta}$ preserves purity.
 \item[(d$\Leftrightarrow$c)] This equivalence is well-known, see
    \eg \cite{keylwerner:lectures}; it is also easily seen using the
    graphical calculus.
  \item[(e$\Leftrightarrow$d)] One direction is trivial; the other
    follows from Lemma~\ref{lem:deltadagid}(c). 
  \item[(g$\Rightarrow$d)] Combining Lemma~\ref{lem:deltadagid}(c)
    with (g) gives that all eigenvalues of $\delta^\dag(\id)$ are 1,
    and hence that $\delta^\dag(\id)$ is (unitarily similar to) the
    identity matrix.
  \item[(f$\Leftrightarrow$e)] Clearly (c), and hence (e), implies
    (f), so it suffices to show (f)$\Rightarrow$(e). If $\delta$ is
    trace-nonincreasing, then 
    $\Tr(a-a\delta^\dag(\id))
     = \inprod{\id-\delta^\dag(\id)}{a}
     = \inprod{id}{a} - \inprod{\id}{\delta(a)}
     = \Tr(a) - \Tr(\delta(a)) \geq 0
    $ for all $a$. If $x_i>0$ is the $i$th eigenvalue of
    $\delta^\dag(\id)$, then by taking
    $a=\mathrm{diag}(0,\ldots,1,\ldots,0)$ with a single 1 in the
    $i$th place, we obtain $x_i-1 \leq 0$. Hence $\delta^\dag(\id)
    \leq \id$.
  \item[(h$\Rightarrow$g)] $\delta^\dag(\id) = \delta^\dag((\sum_{\alpha \in A} \alpha) \tensor
    (\sum_{\beta \in A} \beta)) = \sum_{\alpha,\beta \in A}
    \delta^\dag(\alpha \tensor \beta) = \sum_{\alpha \in A} \alpha$,
    so $\opnorm{\delta^\dag(\id)}=\opnorm{\id}=1$.
  \item[(i$\Rightarrow$h)] Use Plancherel's theorem to get $\id =
    \sum_\alpha \inprod{\id}{\alpha} \alpha = \sum_\alpha \Tr(\alpha)
    \alpha$. Now if $\alpha$ is copyable but not positive
    semidefinite, then by Proposition~\ref{prop:sufficientconditions}(b)
    we have $\Tr(\alpha)=\Tr(|\alpha|)=0$ iff $|\alpha|=0$ iff $\alpha=0$. Hence we
    can take $A=\{\alpha \mid \Tr(\alpha)=1\}$.
  \item[(j$\Rightarrow$i)] If $\alpha$ is positive semidefinite, then by
    assumption (j) we have $\rank(\alpha)=1$, so that $\alpha=U
    \ket{\phi}\bra{\phi} U^\dag$ for some unit vector $\phi$ and unitary
    $U$. Therefore $\Tr(\alpha)=\Tr(\ket{\phi}\bra{\phi})=\inprod{\phi}{\phi}=1$.
 \item[(k$\Rightarrow$j)] if $\alpha\neq 0$ then also
   $\alpha^\dag\alpha\neq 0$. Hence
   $\rank(\alpha)=\rank(\alpha^\dag\alpha)=1$.
 \item[(l$\Rightarrow$k)] follows from
   Proposition~\ref{prop:sufficientconditions}(a).
 \item[(m$\Rightarrow$l)] If $\delta(\id)$ is idempotent, then it is a
   projection by Proposition~\ref{prop:sufficientconditions}(a). By
   Proposition~\ref{prop:sufficientconditions}(c) it commutes with the
   image of $\delta$, and hence is contained in the support projection
   of that image. Hence $\delta(\id)\delta(a)=\delta(a)$ for all
   matrices $a$, so that $\delta$ preserves multiplication by
   Proposition~\ref{prop:sufficientconditions}(c), from which (l)
   follows. 
 \item[(a$\Rightarrow$m)] If $\delta$ is canonical, then $\alpha$ are of
   the form $\ket{i}\bra{j}$, so that $\id=\sum_i \ket{i}\bra{i}$,
   and we obtain $\delta(\id)^2=(\sum_i \delta(\ket{i}\bra{i}))(\sum_j
   \delta(\ket{j}\bra{j})) = \sum_i \ket{ii}\bra{ii} = \delta(\id)$.
 \item[(a$\Leftrightarrow$n)] is just a reformulation in terms of the
   Kraus decomposition; the ancilla space is the one called $Z$ in the
   definition of $\CPM(\cat{C})$ in Section~\ref{sec:preliminaries}.
\end{description}
\end{proof}

\section{Sequentializable quantum protocols}\label{sec:seq}

Up to normalization, classical structures in $\Cat{fdHilb}$ produce a
state that plays an important role in many protocols studied in
quantum information and computation, namely the \emph{maximally
entangled state}. This state is produced by postcomposing the unit of the
classical structure with as many comultiplications as are needed to
get an $m$-party state. Typical examples, frequently encountered in
quantum metrology, are optimal \emph{phase estimation} protocols. They can be
modeled graphically as the following diagram, read left to right:
\begin{align}\label{eq:parallel}
  \vcenter{\hbox{\begin{tikzpicture}[scale=0.4]
	\begin{pgfonlayer}{nodelayer}
		\node [style=none] (0) at (-1, 2) {};
		\node [style=box] (1) at (0, 2) {$\;\;f_m\;\;$};
		\node [style=none] (2) at (1, 2) {};
		\node [style=none] (3) at (-3, 1) {};
		\node [style=dot] (4) at (-2, 1) {};
		\node [style=dot] (5) at (2, 1) {};
		\node [style=none] (6) at (3, 1) {};
		\node [style=none] (7) at (-3.5, 0.5) {$\udots$};
		\node [style=none] (8) at (3.5, 0.5) {$\ddots$};
		\node [style=none] (8a) at (0, -1) {$\vdots$};
		\node [style=none] (9) at (-4, 0) {};
		\node [style=none] (10) at (-1, 0) {};
		\node [style=box] (11) at (0, 0) {$f_{m-1}$};
		\node [style=none] (12) at (1, 0) {};
		\node [style=none] (13) at (4, 0) {};
		\node [style=dot] (14) at (-5, -1) {};
		\node [style=dot] (15) at (5, -1) {};
		\node [style=none] (16) at (-4, -2) {};
		\node [style=none] (17) at (-1, -2) {};
		\node [style=box] (18) at (0, -2) {$\;\;f_2\;\;$};
		\node [style=none] (19) at (1, -2) {};
		\node [style=none] (20) at (4, -2) {};
		\node [style=dot] (21) at (-7, -2.25) {};
		\node [style=dot] (22) at (-6, -2.25) {};
		\node [style=dot] (23) at (6, -2.25) {};
		\node [style=dot] (24) at (7, -2.25) {};
		\node [style=none] (25) at (-5, -3.5) {};
		\node [style=none] (26) at (-1, -3.5) {};
		\node [style=box] (27) at (0, -3.5) {$\;\;f_1\;\;$};
		\node [style=none] (28) at (1, -3.5) {};
		\node [style=none] (29) at (5, -3.5) {};
	\end{pgfonlayer}
	\begin{pgfonlayer}{edgelayer}
		\draw (25.center) to (26.center);
		\draw[bend left=45] (22) to (14);
		\draw (16.center) to (17.center);
		\draw[bend left=45, looseness=1.25] (13.center) to (15);
		\draw[bend right=45] (29.center) to (23);
		\draw (21) to (22);
		\draw (19.center) to (20.center);
		\draw[bend left=45, looseness=1.25] (2.center) to (5);
		\draw[bend left=45, looseness=1.25] (15) to (23);
		\draw (26.center) to (28.center);
		\draw[bend right=45, looseness=1.25] (20.center) to (15);
		\draw[bend left=315] (22) to (25.center);
		\draw[bend left=45, looseness=1.25] (14) to (9.center);
		\draw[bend right=45, looseness=1.25] (4) to (10.center);
		\draw (28.center) to (29.center);
		\draw (3.center) to (4);
		\draw (10.center) to (12.center);
		\draw (23) to (24);
		\draw (17.center) to (19.center);
		\draw (5) to (6.center);
		\draw[bend left=45, looseness=1.25] (4) to (0.center);
		\draw (0.center) to (2.center);
		\draw[bend right=45, looseness=1.25] (14) to (16.center);
		\draw[bend right=45, looseness=1.25] (12.center) to (5);
	\end{pgfonlayer}
      \end{tikzpicture}}} 
\end{align}

In one of most common cases, the diagram is interpreted in the
category $\Cat{fdHilb}$ and the morphisms $f_j$ are all
identical phase gates $e^{-i \phi \sigma_z /2 }$, with unknown phase
$\phi$ and Pauli matrix $\sigma_z$. This entangled protocol produces an estimator $\hat{\phi}$ of
the unknown phase with an uncertainty scaling $\delta \hat{\phi}
\propto 1/m$, which is called the Heisenberg
limit~\cite{giovannetti_quantum_2006,boixo_generalized_2007}. It
follows from the \emph{generalized spider
theorem}~\cite{coeckeduncan:observables} that if the linear maps 
$f_j$ are ``generalized phases'' that are compatible with the classical
structure, then the previous diagram is equivalent to the following one
\begin{align}\label{eq:sequential}
\vcenter{\hbox{\begin{tikzpicture}[scale=0.75]
	\begin{pgfonlayer}{nodelayer}
		\node [style=dot] (0) at (-5, 0) {};
		\node [style=box] (1) at (-3.5, 0) {$\;f_{\pi(1)}\;$};
		\node [style=box] (2) at (-2, 0) {$\;f_{\pi(2)}\;$};
		\node [style=none] (3) at (-1, 0) {};
		\node [style=none] (4) at (-0.05, -0.05) {$\cdots$};
		\node [style=none] (5) at (.9, 0) {};
		\node [style=box] (6) at (2, 0) {$f_{\pi(m-1)}$};
		\node [style=box] (7) at (3.6, 0) {$\,f_{\pi(m)}\,$};
		\node [style=dot] (8) at (5, 0) {};
	\end{pgfonlayer}
	\begin{pgfonlayer}{edgelayer}
		\draw (2) to (3.center);
		\draw (5.center) to (6);
		\draw (7) to (8);
		\draw (0) to (1);
		\draw (6) to (7);
		\draw (1) to (2);
	\end{pgfonlayer}
      \end{tikzpicture}}}
\end{align}
for any permutation $\pi \in S(m)$. This implies
that the entangled parallel protocol~\eqref{eq:parallel} is equivalent to a sequential
protocol~\eqref{eq:sequential} with the same uncertainty scaling. This equivalence has been
studied in frame synchronization~\cite{rudolph_quantum_2003} and clock
synchronization~\cite{boixo_decoherence_2006,higgins_entanglement-free_2007}
between two parties. In quantum computation, the transformation
between sequential and entangled protocols has also been used to study
the computational power of quantum circuits with restricted
length~\cite{moorenilsson:parallel,hoyer_quantum_2003}.

Categorically, the wires and boxes in the diagrams above can also be interpreted
as finite-dimensional Hilbert spaces and quantum maps,
\textit{i.e.}~trace-preserving completely positive maps. These are the
maps quantum information theory concerns~\cite{keylwerner:lectures}.
For quantum computation, it makes sense to relax to
trace-nonincreasing completely positive
maps~\cite{selinger:superoperators}. Neither form a dagger  
compact category in themselves, but we may regard them as living within
$\CPM(\Cat{fdHilb})$. The equivalence between~\eqref{eq:parallel}
and~\eqref{eq:sequential} holds unabated for quantum  
maps $\cp{f}$ in $\CPM(\Cat{fdHilb})$, as long as they are compatible
with the classical structure. 

This interpretation has the advantage that general quantum maps are
able to model \emph{noisy} estimation protocols. Hence we can now,
fully generally, address the question of whether the parallel protocol
with maximally entangled states is more fragile in its response to
noise than the corresponding sequential protocol. On the one hand,
physical intuition tells us that this might be the case, given that
entanglement is considered to be a very delicate resource in
general. On the other hand, the equivalence between both diagrams
derived from the categorical machinery means that for certain kinds of
noise the entangled protocol is not more fragile than the sequential
one. This generalizes the equivalences of clock synchronization
protocols under noise that have been studied for two-dimensional
Hilbert space in~\cite{boixo_decoherence_2006}. By
Theorem~\ref{thm:equivalentconditions}, classical structures are
quantum maps if and only if they are canonical, and hence correspond
to an orthonormal basis. The only question left is which quantum maps
are compatible with that basis. We now give without proof the explicit
form of such maps, referring to~\cite{boixoheunen:prl} for details.
% The same point of view might have interesting consequences, yet unexplored, in e.g.~conformal
% or topological quantum field theories~\cite{kock:frobenius,
%   baez:quandaries, atiyah:tqft, segal:cft}, and geometric
% quantization~\cite{landsman:quantization}.
 
Being a completely positive operator, a quantum map can be written as
$\rho \mapsto \sum_s b_s \rho b_s^\dagger$, where 
$b_s$ are the Kraus operators. 
Denote the $n$th roots of unity by
$
  \omega_j = e^{-i 2 \pi j/n}
$,
where $n$ is the dimension of the underlying Hilbert space. The
quantum maps that are compatible with a classical structure are
precisely those which can be expressed with Kraus operators of the form
\[
  b_s = \sqrt{r_s} \sum_j e^{-i (\phi_j + 2 \pi j s /n)} \ket j \bra j,
\]
where $\{\ket j\}$ is the basis defined by the classical structure,
$\phi_j$ define arbitrary phase rotations, and $r_j$ are positive
constants parametrizing general dephasing noise, satisfying $\sum_j r_j=1$.  
Maps without dephasing (\textit{i.e.}~pure rotations) are obtained by choosing
$r_s = \delta_{s,0}$. 

Thus, when the above quantum metrology protocols are modeled
with such maps, usage of maximally entangled states is equally
robust as, and has equal uncertainty scaling to, the corresponding sequential version.

% bibliography
\bibliographystyle{eptcs}
\bibliography{sequentializable-qpl}

\end{document}